\documentclass{article}
\usepackage{kr}

\usepackage{times}
\usepackage{soul}
\usepackage{url}
\usepackage[hidelinks]{hyperref}
\usepackage[utf8]{inputenc}
\usepackage[small]{caption}
\usepackage{graphicx}
\usepackage{amsmath}
\usepackage{amsthm}
\usepackage{booktabs}
\usepackage{algorithm}
\usepackage{algorithmic}
\newcommand{\stable}{{\mathit{stb}}}
\newcommand{\stb}{\stable}
\newcommand{\adm}{\mathit{ad}}
\newcommand{\prf}{\mathit{pr}}

\newcommand{\grd}{\mathit{gr}}
\newcommand{\com}{\mathit{co}}

\newcommand{\myparNS}[1]{\noindent{\bfseries #1.}}
\newcommand{\mypar}[1]{\smallskip\myparNS{#1}}
\usepackage{enumitem}

\graphicspath{{./figures/}}

\newtheorem{theorem}{Theorem}
\newtheorem{definition}{Definition}

\title{%Title options: Layered AF Visualization\\
Visualizing Extensions of Argumentation Frameworks as Layered Graphs}

\author{%
Martin N\"ollenburg$^1$\and
Christian Pirker$^1$\and
Anna Rapberger$^2$\and
Stefan Woltran$^1$\and
Jules Wulms$^3$\\
\affiliations
$^1$TU Wien\quad
$^2$Imperial College London\quad
$^3$TU Eindhoven\\
\emails
noellenburg@ac.tuwien.ac.at, e11908083@student.tuwien.ac.at, a.rapberger@imperial.ac.uk, woltran@dbai.tuwien.ac.at, j.j.h.m.wulms@tue.nl
}

\begin{document}

\maketitle

\begin{abstract}
The visualization of argumentation frameworks (AFs) is crucial for enabling a wide applicability of argumentative tools. However, their visualization is often considered only as an accompanying part of tools for computing semantics
and standard graphical representations are used. 
We introduce
a new visualization technique that draws an AF, together with an extension (as part of the input),
as a 3-layer graph layout. Our technique supports the user
to more easily explore the visualized AF, better understand extensions, and verify algorithms for
computing semantics.
To optimize the visual clarity and aesthetics of this layout, we propose to minimize edge crossings in our
3-layer drawing.  We do so by an exact ILP-based approach, but also
propose a fast heuristic pipeline. Via a quantitative evaluation, we show that the heuristic is feasible even for large instances, while producing at most twice as many crossings as an optimal drawing in most cases.  
\end{abstract}

\section{Introduction}
Computational models of argumentation provide methods for reasoning under uncertainty; they are among the core formalisms of knowledge representation and reasoning.
\emph{Abstract argumentation}~\cite{Dung95} is one of the key approaches to argumentative reasoning, with a broad range of applications in legal reasoning, medical sciences, and e-democracy, see, e.g., ~\cite{AtkinsonBGHPRST17,arguHandbook}. 
In Dung's argumentation frameworks (AFs), arguments are considered abstract; reasoning is entirely based on the conflicts between arguments. 
Argumentation semantics provide criteria for argument acceptance; 
they can be employed to determine sets of jointly acceptable arguments (extensions). 
The visualization of AFs  is crucial for the deployment of formal argumentation.
AFs are transparent and explainable by design; their graph-based structure allows for an intuitive and easy-to-understand visual representation of conflicting scenarios. For instance,
\citeauthor{BesnardDDL22} 
\shortcite{BesnardDDL22} 
utilize the visualization of subgraphs of AFs for argumentative explanations. 
However, the design of powerful and flexible visualization tools is often considered secondary, in particular when it comes to depict the result
in terms of extensions.
%AF visualization is usually only an accompanying part of tools for computing semantics.
To unleash their full potential, it is essential to provide customized visualizations of AFs rather than applying general-purpose graph drawing techniques.

In this work, we develop genuine methods for the graphical representation of AFs. We introduce a new visualization technique that draws an AF, together with an extension (as part of the input), as a 3-layer graph layout. Our technique supports the user to more easily explore the visualized argumentation framework, better understand extensions, and graphically verify algorithms for computing semantics. 
%Unlike other tools, our goal is to innovate on the visualization of AFs and semantics, independent of the semantics computation. 

Our 3-layer approach allows us to highlight multiple properties of the provided extension: The layers make it easy to identify that the extension is conflict-free and admissible. Using different colors we further highlight various witness sets for admissibility and to show whether an extension is maximal. To optimize the visual clarity and aesthetics of this visualization, we propose to minimize edge crossings in our 3-layer drawing. We show that minimizing crossings under additional constraints imposed by the highlighting cannot be achieved by standard crossing minimization approaches for layered graphs. To optimize our visualizations, we therefore propose a fast heuristic pipeline and an exact %ILP-based 
integer linear programming approach, the latter being feasible only for small instances. We quantitatively evaluate the run times and number of crossings via a prototype implementation.\footnote{Demo prototype at \url{https://christianlinuspk.github.io/ArguViz/}; code available at \url{ https://github.com/ChristianLinusPK/ArguViz/}\label{note:demo}}

\begin{figure*}
    \centering
    \includegraphics[width=.3\textwidth]{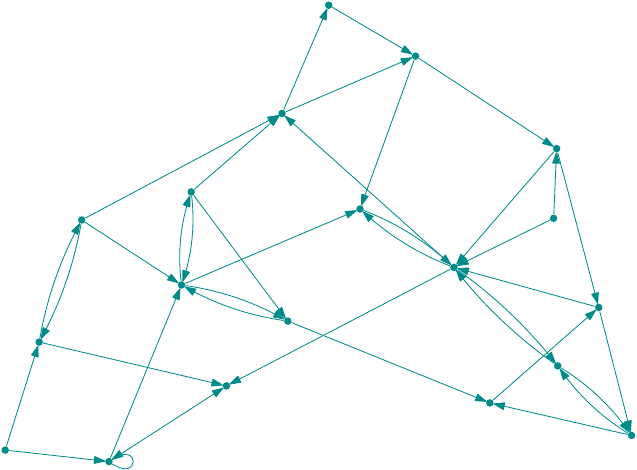}
    \hfill
    \includegraphics[width=.3\textwidth]{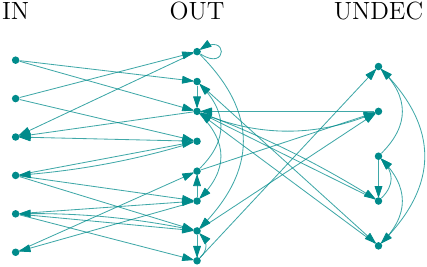}
    \hfill
    \includegraphics[width=.3\textwidth]{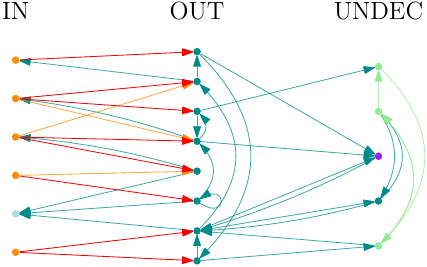}
    \caption{The same AF visualized using a force-directed layout (left), and layered AF drawings; basic (middle) and optimized (right).}
    \label{fig:design-example}
\end{figure*}

\paragraph{Related Work}
The efficient computation of extensions has received significant attention in the literature; see e.g., \cite{NiskanenJ20a,DvorakGRWW20,DVORAK201453,bistarelli2011conarg},
also witnessed by the biennial International Competition on Computational Models of Argumentation (ICCMA)
\cite{DBLP:conf/kr/JarvisaloLN23}.
%\cite{ThimmV17,GagglLMW20,LagniezLMR20,BistarelliKST21,DBLP:conf/kr/JarvisaloLN23}.
%1) paragraph on AFs and computing extensions (Anna)
The representation of argument networks, however, is often only an accompanying function of extension retrieval.
%ee \url{https://cmna.csc.liv.ac.uk/CMNA7/papers/Schneider.pdf}
%
%Example mentioned by Stefan: 
%
%\citeauthor{BesnardDDL22} (\citeyear{BesnardDDL22}) utilize the visualization of subgraphs as explanations for argument acceptance.
%
Basic argument visualization strategies independent of the computation of the extensions, find applications in the OVA+ tool for analyzing arguments and debates~\cite{reed2014ova,lawrence2012aifdb} or in the Argunet system, designed for collaborative argumentation analysis %by Gregor Betz et al~
\cite{Argunetarticle}.
Moreover, the visualization tool NEXAS~\cite{DGKMRY2022} provides means to navigating and exploring argumentation solution spaces, but employs other visualization methods like scatterplots and matrix representations.

%TODO: check Christian's thesis and the FWF proposal 

% \newpage
Drawing directed graphs on two or more layers is a well-known drawing 
paradigm in graph drawing~\cite{stt-mvuhss-81,hn-hda-13}. For a fixed assignment of vertices to layers (which can be computed or given as part of the input) the main task studied for optimizing the readability of the resulting drawing is to minimize the number of edge crossings~\cite{p-wagehu-97}. In so-called proper layered drawings, edges connect vertices of adjacent layers only and the number of crossings depends only on the permutations of vertices of adjacent layers. This crossing minimization problem is NP-hard, even for a 2-layer drawing, in which the permutation of one layers is fixed~\cite{EadesW94}. Several heuristics, approximations, and exact integer programming (ILP) methods for two or more layers are known~\cite{stt-mvuhss-81,ew-ecdbg-94,DBLP:conf/gd/JungerLMO97} and their performance has been experimentally evaluated~\cite{jm-2scmpeha-97}. While some variants of layered crossing minimization in the presence of ordering constraints have been studied~\cite{Forster04}, no previous results are known for the specific constraints needed when drawing AFs with extensions and highlighting their structural properties as done in this paper.

%3) graph drawing algorithms for layered graphs and stack layouts

%classic Sugiyama framework~\cite{stt-mvuhss-81}

%overview chapter on hierarchical layout~\cite{hn-hda-13}

%hardness of two-layer crossing minimization~\cite{EadesW94}

%multi-layer crossing minimization~\cite{DBLP:conf/gd/JungerLMO97} (Our ILP is adapted from this)

%heuristics~\cite{ew-ecdbg-94} and exact methods~\cite{jm-2scmpeha-97}

%constrained version~\cite{Forster04}

\section{Preliminaries}
We fix a non-finite background set ${U}$. An argumentation framework (AF) \cite{Dung95} is a directed graph $F = (A,R)$ where $A\subseteq{U}$ represents a set of arguments and $R\subseteq A\times A$ models \textit{attacks} between them. 
We let $E^+_F = \{ x\in A \mid E\text{ attacks }x \}$ for a set $E\subseteq A$.
%In this paper we consider finite AFs only. %and we use $\m{F}$ for the set of all these graphs. 
%For a given $\F = (B,S)$ we let $A(\F) = B$ and $R(F) = S$.
%For $U\subseteq A$ we define the restriction of $\F$ to $U$ as usual, i.e.\ $F\!\!\downarrow_U = (A\cap U , R\cap ({U\times U}))$.
For two arguments $x,y\in A$, if $(x,y)\in R$ we say that $x$ \textit{attacks} $y$ as well as $x$ \textit{attacks} (the set) $E$ given that $y\in E\subseteq A$.
%We frequently use the so-called \textit{range} of a set $E$ defined as $E^\oplus_F = E\cup E^+_F$ where $E^+_F = \{ a\in A \mid E\text{ attacks }a \}$.
%The \textit{$E$-reduct of $\F$} is the AF $F^E = (E^*, R \cap (E^* \times E^*))$ where $E^* = A \setminus E^\oplus$.
%This means, $\F^E$ is the subframework of $\F$ obtained by removing the range of $E$.%, \ie $F^E = F\!\!\downarrow_{A\setminus E^\oplus_F}$. 
%
%Given AF $G=(B,S)$, the union $F\cup G$ is defined as  $\left(A \cup B, R\cup S\right)$; we say $G$ \emph{expands} $F$. Likewise, $F\setminus G=\left(A \setminus B, (R\setminus S)\cap (A\setminus B)^2\right)$; we write $F\setminus B$ if $S=\emptyset$.

\newcommand{\F}{F}

A set $E\subseteq A$ is \emph{conflict-free} in $\F$ %(for short, $E\in\cf(\F)$) 
iff for no $x,y\in E$, $(x,y)\in R$. 
%We say 
$E$ \textit{defends} an argument $x$ if 
$E$ attacks each attacker of $x$.
%any attacker of $x$ is attacked by some argument of $E$.
A conflict-free set $E$ is \emph{admissible} in $\F$ ($E\in \adm(\F)$) iff it defends all its elements.
%(we also say, $E$ defends itself).
A \emph{semantics} is a function $\sigma:\mathcal{F}\to 2^{2^{U}}$ with $F\mapsto\sigma(F)\subseteq 2^A$.
This means, given an AF $\F = (A,R)$ a semantics returns a set of subsets of $A$. These subsets are called $\sigma$-\emph{extensions}.
In this paper we consider so-called 
%\textit{naive},
%\textit{admissible}, 
\textit{complete}, 
\textit{grounded},
\textit{preferred}, 
and \textit{stable}
%\textit{semi-stable}
%and \textit{stage}
semantics (abbr. 
%$\nav$, 
%$\adm$,
$\com$, 
$\grd$, 
$\prf$,  
$\stb$).
\begin{definition} \label{def:extsem}
	Let $\F = (A,R)$ be an AF and $E\in\adm(\F)$. %
	%\begin{enumerate}
	%\begin{itemize}
	%	\item 
  $E\in\com(\F)$ iff $E$ contains all arguments it defends;
		%\item 
  $E\in\grd(\F)$ iff $E$ is $\subseteq$-minimal in $\com(\F)$;
	%	\item 
 $E\in\prf(\F)$ iff $E$ is $\subseteq$-maximal in $\com(\F)$;
	%	\item 
 $E\in\stb(\F)$ iff $E^+=A\setminus E$.	
	%\end{itemize}
%		$E\in \semi(\F)$ iff $E\in \com(\F)$ and there is no $D\in\com(\F)$ with $E\cup E^+_F\subset D\cup D^+_F$.
	%\end{enumerate}
\end{definition}
We make use of argument labelings to visualize extensions~\cite{BaroniCG18}. 
Let $E$ be a $\sigma$-extension. 
An argument $a\in A$ is labelled \textsc{in} if $a\in E$; \textsc{out} if there is $b\in E$ s.t.\ $(b,a)\in R$; and  \textsc{undec} otherwise.

A \emph{proper $k$-layer drawing} of $F=(A,R)$ consists of (i) a partition of $A$ into $k$ subsets $A_1, \dots, A_k$ called \emph{layers} such that for each edge $(x,y) \in R$ either $x,y$ belong to the same layer $A_i$ (\emph{intra-layer edge}) or they belong to adjacent layers $A_i, A_{i+1}$ (\emph{proper edge}), (ii) a permutation $\pi_i$ for each layer $A_i$ mapping each argument $a \in A_i$ to a position $(i,\pi_i(a))$ on a set of vertical lines in the plane,  (iii) a straight-line segment connecting the endpoints of each proper edge, and (iv) an arc to the right of each layer connecting the endpoints of each intra-layer edge. We consider the cases of $k=2$ and $k=3$, see Fig.~\ref{fig:design-example} (middle) for a 3-layer example.

% TODO (Martin): put here basic preliminaries on layered graph drawings: map vertices to k vertical layers (here k=2 or 3), (proper) edges link between adjacent layers, (level) edges link within a layer; the former are straight, the latter are semi-arcs; the drawing objective is to determine vertex orders within each layer, which already combinatorially determines all crossings (inversions of edges are crossings)

\section{Visualization Design} 
In this section we elaborate on the design of our visualization. We start by introducing and motivating the design decisions that underlie our layered AF drawings. Then we elaborate on how we optimize visual clarity and aesthetics. Finally, we prove a few properties of our drawings.

\mypar{Design decisions} An important distinction between our AF visualization and existing works is that it is not solely a graphical interface for a solver that computes semantics. To the contrary, our technique for an AF $F=(A,R)$ requires that the semantics have already been computed, and hence an extension~$E\subseteq A$ is also provided as input. We visualize both an AF~$F$, as well as an extension~$E$. %to allow the user to explore the argumentation framework, build understanding of the framework and the extension, and verifying properties of the extension.
For that, we propose the following basic design decisions (DDs).

\begin{enumerate}
    \item\label{dd:graph} Visualize the AF~$F=(A,R)$ as a layered graph drawing.
    \item\label{dd:vertices} All arguments~$A$ are represented by vertices that are partitioned over (at most) three layers according to their argument labelings: the \textsc{in}-, \textsc{out}, and \textsc{undec}-layers.
    \item\label{dd:edges} The attacks~$R$ between arguments are represented as directed edges following proper 3-layer drawing standards.
\end{enumerate}

See Fig.~\ref{fig:design-example} (middle) for an example of a visualization adhering to these basic design decisions. DDs~\ref{dd:graph} and~\ref{dd:edges} are fairly standard; AFs are often visualized as graphs, in which arguments and attacks are represented as vertices and (directed) edges, respectively. However, DD~\ref{dd:vertices} ensures that our visualization is distinctively different from existing work, by incorporating the given extension~$E$ into the layout of the graph. Compared to standard techniques (see Fig.~\ref{fig:design-example} (left)), this emphasizes the arguments in the extension (the \textsc{in}-layer) and structures the remaining arguments with respect to the extension (\textsc{out}- and \textsc{undec}-layers). Note that we can easily verify that the set~$E$ is \emph{conflict-free}, by the absence of edges between vertices in the \textsc{in}-layer. 
%Moreover, DD~\ref{dd:edges} can be satisfied since there are no attacks between arguments in the \textsc{in} and \textsc{undec}-layers. 
Building on this basic design, we introduce further design decisions for attacks.

\begin{enumerate}
    \setcounter{enumi}{3}
    \item\label{dd:in-attacks} Attacks from \textsc{in} to \textsc{out} have a distinct color.
    \item\label{dd:red} For each argument in \textsc{out}, highlight one incoming edge.
    \item\label{dd:odd-cycle} Odd cycles in \textsc{undec} are highlighted.
\end{enumerate}

For an example of these design decisions, see Fig.~\ref{fig:design-example} (right): We chose orange as the color to distinguish attacks from \textsc{in} to \textsc{out}; we highlight one edge for each argument in \textsc{out} in red and odd cycles in \textsc{undec} using (light) green. 

DD~\ref{dd:in-attacks} and~\ref{dd:red} facilitate the visual communication of whether the extension is \emph{admissible}. %Admissible extensions have the property that every argument that attacks an argument in \textsc{in}, is also attacked by an argument in \textsc{in}. 
To see this, note that, by definition, edges starting in the \textsc{in}-layer point to arguments in \textsc{out}, and edges between \textsc{in} and \textsc{undec} cannot exist for admissible extensions. Hence coloring and highlighting to verify admissibility can be restricted to edges between \textsc{in} and \textsc{out}. Specifically, DD~\ref{dd:red} is implemented by coloring certain edges red, which ensures a strong visual presence of a witness for admissibility. Moreover, DD~\ref{dd:in-attacks} allows the user to further explore and understand the admissibility property. For example, the user can look for other witnesses among the (orange) colored edges, or find out that admissibility hinges on a single attack, the deletion of which would no longer make this extension admissible.

Finally, DD~\ref{dd:odd-cycle} focuses on the \textsc{undec}-layer. Arguments in \textsc{undec} can potentially be added to~$E$. This is especially interesting when arguments in \textsc{out} attack arguments in \textsc{undec}: the attacked arguments in \textsc{undec} are then already defended, and need to be added to~$E$ to obtain a complete extension. 
%However, consider a directed cycle in the AF, and observe that for an admissible set, no two consecutive arguments in a directed cycle can be part of~$E$. 
However, arguments in an isolated cycle of odd length cannot be part of a complete extension, by definition. DD~\ref{dd:odd-cycle} produces visual witnesses for arguments that cannot be added to a complete extension~$E$, unless attacks outside the odd cycle interfere. Thus, this can trigger further investigation into whether~$E$ is or can become a complete extension.

We round out the set of design decisions with a few design decisions for arguments.

\begin{enumerate}
    \setcounter{enumi}{6}
    \item\label{dd:attacking-arguments} Attacking arguments are colored according to their attack.
    \item\label{dd:attacked-in} Arguments in \textsc{in} that do not attack are highlighted.
    \item\label{dd:attacked-undec} Arguments in \textsc{undec} that are not attacked by other arguments in \textsc{undec} are highlighted --- DD~\ref{dd:attacking-arguments} takes priority.
\end{enumerate}

Again see Fig.~\ref{fig:design-example} (right): Attacking arguments in \textsc{in} match the (orange) color of the edges (even when highlighted in red) and arguments in \textsc{undec} match the (light) green color of their odd cycle edges. Arguments in \textsc{in} that do not attack are highlighted in (light) blue, while arguments in \textsc{undec} that do not attack are highlighted in purple.

DD~\ref{dd:attacking-arguments} ensures a coherent look and makes it easier to trace attacks back to the attacking argument. DD~\ref{dd:attacked-in} again helps in verifying admissibility of the extension: arguments in \textsc{in} that do not attack are not defending any arguments and hence must be defended by another argument.

Similar to DD~\ref{dd:odd-cycle} for the \textsc{in}-layer, DD~\ref{dd:attacked-undec} focuses on the \textsc{undec}-layer. DD~\ref{dd:attacked-undec} ensures that arguments that are already defended are highlighted. For admissible extensions, these attacks may be added to \textsc{in}.

\mypar{Optimization} Given these design decisions, we want to produce visualizations that are as clear and legible as possible. Consider Fig.~\ref{fig:design-example} once more, and observe that on the right, vertex orders are optimized to minimize edge crossings. Minimizing edge crossings is a well-known %, but generally NP-hard~\cite{gj-cnn-83} 
optimization criterion for graph drawings that ensures visual clarity and aesthetically pleasing visualizations~\cite{p-wagehu-97}.

We add non-standard constraints to our crossing minimization. Remember that the order of the vertices in the layers determines the crossings, and that we have intra-layer and proper edges. We argue that the crossings of different edges do not affect our visualization equally: the proper edges between the \textsc{in}- and \textsc{out}-layer are crucial for the verification and exploration of extension properties. We therefore place more emphasis on minimizing their crossings. 
\begin{enumerate}[label=\Roman*.]
    \item We introduce the \emph{red-edge constraint} (REC): No two highlighted edges for DD~\ref{dd:red} may cross.
    \item We  minimize (1) the number of edge crossings between the \textsc{in}- and \textsc{out}-layer as primary objective, and (2) all other edge crossings as secondary objective.
\end{enumerate}
Let \textsc{AFCrossingMinimization} be the optimization problem that asks to compute a layered AF drawing that minimizes crossings according to II. We consider variants of this problem under the REC (\textsc{AFCrossMinREC}) and disregarding the REC (\textsc{AFCrossMin}).

\subsection{Properties of Optimal Layered AF Drawings}
\label{sec:properties}

In this section we investigate the relation between \textsc{AFCrossMin} and \textsc{AFCrossMinREC}. We consider very simple AFs that have only \textsc{in}- and \textsc{out}-layers, only proper edges between \textsc{in} and \textsc{out}, and small extensions of size two or three. We first show that solving \textsc{AFCrossMinREC} and \textsc{AFCrossMin} on certain instances results in equal crossing counts (see Fig.~\ref{fig:theorem1}).

\begin{theorem}\label{thm:afcrossmin-equal}
    Layered AF drawings of an AF $F=(A,R)$ with stable extension~$E$, with $|E|=2$, obtained by solving \textsc{AFCrossMinREC} and by solving \textsc{AFCrossMin} have the same number of total edge crossings.
\end{theorem}
\begin{proof}
Consider an arbitrary AF $F=(A,R)$ with a stable extension~$E$ of size two. Let $u,v\in E\subseteq A$ be the arguments in the extension, and observe that the theorem holds when $A\setminus E = \emptyset$. Thus assume that $A$ contains other arguments.

Without loss of generality we assume that $u$ is drawn above $v$ in the \textsc{in}-layer of the layered AF drawing of~$F$. Since $E$ is a stable extension, all arguments in $A\setminus E$ will be drawn in the \textsc{out}-layer. We first observe that, drawing all degree-1 arguments adjacent to~$u$ at the top of the \textsc{out}-layer can prevent them from incurring any crossings, and symmetrically, all degree-1 neighbors of~$v$ can be drawn at the bottom of the \textsc{out}-layer. This leaves only the arguments that are adjacent to both~$u$ and~$v$, and let there be $k$ of such argument. These are (necessarily) drawn in the middle of the \textsc{out}-layer, between the degree-1 arguments (see Figure~\ref{fig:theorem1}).

\begin{figure}[t]
    \centering
    \includegraphics[width=.19\textwidth]{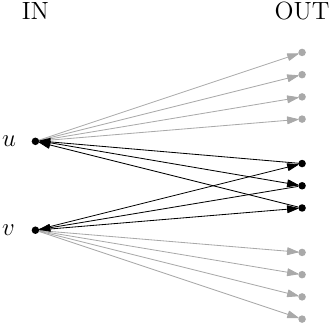}
    \qquad
    \includegraphics[width=.19\textwidth]{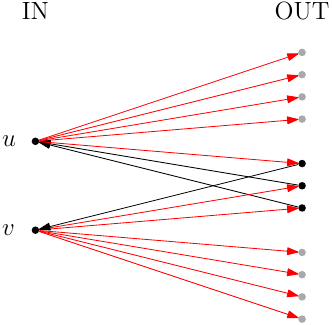}
    \caption{When $|E|=2$, solving \textsc{AFCrossMin} (left) and \textsc{AFCrossMinREC} (right) results in equal crossing counts.}
    \label{fig:theorem1}
\end{figure}

Each of these $k$ arguments will incur $k-1$ edge crossings regardless of their ordering: Take an arbitrary argument $a\in A$ in this set of $k$ arguments and let the $k$ arguments be drawn in arbitrary order. The edge $(u,a)$ (or $(a,u)$) will cross all edges between $v$ and the degree-2 arguments above~$a$. Symmetrically, $(v,a)$ (or $(a,v)$) crosses all edges between $u$ and the degree-2 arguments below~$a$. This leaves us with $\frac{k\cdot (k-1)}{2}$ crossings in total for \textsc{AFCrossMin} (because we double count crossing when counting $k-1$ crossings for each of the $k$ arguments).

To complete the proof, observe that the same approach can minimize crossings for \textsc{AFCrossMinREC}: Place degree-2 vertices between the degree-1 arguments to incur $\frac{k\cdot (k-1)}{2}$ total edge crossings. Since any ordering of the degree-2 arguments results in the same number of crossings, we can sort the degree-2 arguments depending on whether the incident red edges originates from~$u$ or~$v$. We place arguments with red edges coming from~$u$ above those with red edges from~$v$ to satisfy the REC (see Figure~\ref{fig:theorem1}).
\end{proof}

Theorem~\ref{thm:afcrossmin-equal} may seem to suggest that standard methods for crossing minimization can be used to solve \textsc{AFCrossMinREC}. However, we can also prove that for some instances the optimal solutions to \textsc{AFCrossMinREC} and \textsc{AFCrossMin} do not coincide (see Fig.~\ref{fig:theorem2}). We therefore need new methods to solve \textsc{AFCrossMinREC}.% and we introduce two such methods in the next section.

\begin{theorem}
    There exists an AF $F=(A,R)$ with extension~$E$, with $|E|=3$, for which solving \textsc{AFCrossMinREC} results in more edge crossings than solving \textsc{AFCrossMin}.
\end{theorem}
\begin{proof}
Consider an arbitrary AF $F=(A,R)$ with a stable extension~$E$ of size three. Let $u,v,2\in E\subseteq A$ be the arguments in the extension, and let $A\setminus E$ contain only degree-1 and degree-2 arguments. See Figure~\ref{fig:theorem2} (left) for an example the drawing constructed in the following paragraphs.

Without loss of generality that $u$, $v$ and $w$ are drawn in this order from top to bottom in the \textsc{in}-layer of the layered AF drawing of~$F$. Since $E$ is a stable extension, all arguments in $A\setminus E$ will be drawn in the \textsc{out}-layer. Drawing all degree-1 arguments adjacent to~$u$ at the top of the \textsc{out}-layer can prevent them from incurring any crossings, and symmetrically, all degree-1 neighbors of~$w$ can be drawn at the bottom of the \textsc{out}-layer. The degree-1 neighbors of~$v$ will be drawn between the other degree-1 arguments, so that the degree-1 arguments do not incur any pairwise crossings.

\begin{figure}[t]
    \centering
    % \hfill
    % \includegraphics[width=.19\textwidth]{not-afcrossmin}
    \hfill
    \includegraphics[width=.19\textwidth]{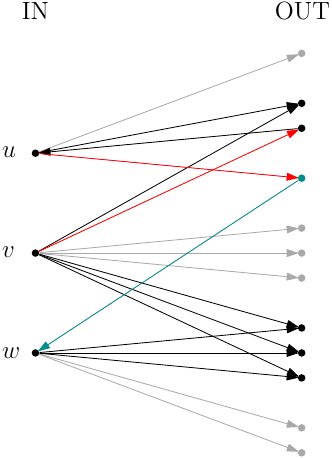}
    \qquad
    \includegraphics[width=.19\textwidth]{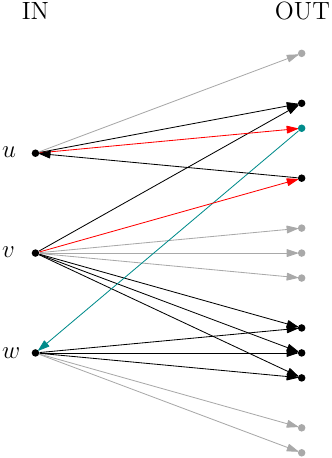}
    \caption{When $|E|=3$, solving \textsc{AFCrossMin} (left) may result in fewer crossings than solving \textsc{AFCrossMinREC} (right).}
    \label{fig:theorem2}
\end{figure}

Next we consider the degree-2 arguments adjacent to~$v$. By placing the degree-2 arguments adjacent to $u$ and $v$ between their respective degree-1 neighbors, we prevent any crossings with the edges incident with the degree-1 neighbors. We can do a symmetric construction for $v$ and $w$, Let there be $k_{u}$ degree-2 arguments adjacent to $u$ and $v$ and $k_{w}$ degree-2 arguments adjacent to $v$ and $w$. The construction so far then incurs $\frac{k_{uv}\cdot (k_{uv}-1)}{2} + \frac{k_{vw}\cdot (k_{vw}-1)}{2}$ crossings.

The only arguments left are the degree-2 arguments incident with $u$ and $w$. Assume there is only one such argument~$a\in A$. To prevent crossings with the degree-1 neighbors of $u$ and $w$, we place~$a$ between these degree-1 neighbors in the ordering of the \textsc{out}-layer. 

Now observe that all edges incident with $v$ are crossed by at least one of the two edges incident with~$a$: All edges incident with~$v$ to arguments above~$a$ cross $(u,a)$ (or $(a,u)$), and all edges incident with~$v$ to arguments below~$a$ cross $(w,a)$ (or $(a,w)$). Furthermore, if $a$ is placed above a neighbor of both $u$ and $v$, then $(w,a)$ (or $(a,w)$) also crosses the edge incident with $u$, and symmetrically when $a$ is placed below a neighbor of both $v$ and~$w$. Hence, ideally $a$ is placed between the other sets of degree-2 arguments, such that the edges incident with $a$ incurs exactly one crossing with each edge incident with~$v$. Let~$E_v$ be the edges incident with~$v$, then solving \textsc{AFCrossMin} can result in at most $\frac{k_{uv}\cdot (k_{uv}-1)}{2} + \frac{k_{vw}\cdot (k_{vw}-1)}{2} + |E_v|$ crossings.

To complete the proof, observe that the direction of the edges can be such that, for a degree-2 neighbor incident with~$v$ the only incoming edge originates at~$v$. Thus, this edge will be chosen as a red edge. Assume without loss of generality that this is true for a neighbor of $u$ and $v$. Now the only incoming edge of $a$ can also originate from~$u$. To satisfy the REC, $a$ must be placed above the previously considered degree-2 neighbor of~$u$ and $v$ (see Figure~\ref{fig:theorem2} (right)). As explained in the previous this results in at least one additional crossing, compared to a placement of~$a$ between the other sets of degree-2 neighbors. In this case, solving \textsc{AFCrossMinREC} results in at least $\frac{k_{uv}\cdot (k_{uv}-1)}{2} + \frac{k_{vw}\cdot (k_{vw}-1)}{2} + |E_v| + 1$ crossings.
\end{proof}

\section{Algorithms}
In this section we introduce two methods to solve \textsc{AFCrossMinREC}: An ILP-based exact approach, and a heuristic pipeline; both are justified for an NP-hard problem.

\mypar{Integer linear programming} We have adapted an ILP for crossing minimization in multi-layer graph drawings~\cite{DBLP:conf/gd/JungerLMO97}. We use Boolean variables for each pair of arguments in a layer, to model their relative order in the permutation. Using so-called transitivity constraints we ensure that these variables produce a valid total order. Furthermore, we partition the edges into four sets~$E_1,\ldots,E_4$ of edges between \textsc{in} and \textsc{out}, within \textsc{out}, between \textsc{out} and \textsc{undec}, and within \textsc{undec}, respectively. We consider crossings only between edges within one such set, and employ Boolean crossing variables~$c_{ijkl}^n$ for every pair of edges $(i,j),(k,l)$ within a set~$E_n$. These variables are set via constraints which differ between intra-layer and proper edge crossings. Our objective function multiplies the crossing variables with constants $w_{ijkl}^n$ as crossing weights. 
\begin{equation*}
    \text{minimize} \sum_{n=1}^4\sum_{(i,j),(k,l)\in E_n} w_{ijkl}^n \cdot c_{ijkl}^n
\end{equation*}
By default this weight is set to~$1$, but for $n=1$ we set the weight to the total number of possible crossings in $E_2,E_3,E_4$ to strictly prioritize crossing minimization in~$E_1$ according to condition~II in the previous section. 

Red edges are selected via boolean variables~$r_{ij}$ for each edge~$(i,j)$ from \textsc{in} to \textsc{out}. Only one $r_{ij}$ can be set for each~$j$ in~\textsc{out}, and dedicated constraints forbid crossings.%Note that some of the edge sets may be empty, for example $E_3$ and $E_4$ are empty for stable extensions ($\textsc{undec} = \emptyset$). This removes some variables and constraints. 

\smallskip
\noindent\emph{The full ILP is described in Appendix~\ref{app:ILP}.}
\smallskip

Modern ILP solvers (e.g. Gurobi) may take over 30 minutes to find an optimal solution using the described ILP (see Section~\ref{sec:evaluation}). Since we intend for our layered AF drawings to be used in interactive settings, we also introduce a faster (albeit often suboptimal) heuristic alternative.

\mypar{Heuristic pipeline} Our heuristic pipeline for solving \textsc{AFCrossMinREC} consists of three main components:
\begin{itemize}
    \item Selecting red edges.
    \item Iteratively applying (variants of) the \emph{barycenter method} (see below) to enforce the REC and minimize crossings. 
    \item Employing a local search to improve red-edge selection.
\end{itemize}

When selecting edges to be highlighted between the \textsc{in}- and \textsc{out}-layer, we employ two strategies (A and B). Strategy A greedily selects as many edges originating from the same argument in~\textsc{in} as possible, ensuring that only few arguments in~\textsc{in} are incident with red edges. %This also makes it easy to adhere to the REC, since the edges originating from the same argument in~\textsc{in} cannot cross, while arguments in~\textsc{out} with incoming red edges from different arguments in~\textsc{in} may not interleave in their ordering. 
Strategy B tries to disperse the highlighted edges over as many arguments in~\textsc{in} as possible, by iteratively selecting (and removing) maximum matchings of edges from~\textsc{in} to~\textsc{out}.

The (default) barycenter method reorders one layer, relative to a static adjacent layer, by placing vertices in the barycenter of its neighbors~\cite{stt-mvuhss-81}. We use this method to ensure the solution adheres to the REC and secondarily to minimize crossings. We use variants of the barycenter method in four ways: First, we reorder~\textsc{out}, such that arguments in~\textsc{out} with red edges originating from the same \textsc{in}-argument become consecutive. Second, we reorder each subset of~\textsc{out} with red edges from the same \textsc{in}-argument, minimizing crossings. Third, we reorder~\textsc{in}-arguments incident with red edges, to remove red-edge crossings. Fourth, we reorder all of~\textsc{in} to minimize overall crossings, while preventing new red-edge crossings.

The local search considers for each argument in~\textsc{out} all its incoming edges from~\textsc{in} and tries to replace the selected red edge, by one of the incoming edges that was not selected. Then the sequence of barycenter methods is run, and the result is compared to the drawing before the red-edge swap. If the number of crossings is reduced, the new drawing is maintained, otherwise we revert to the pre-swap drawing.

Our heuristic pipeline now works as follows. We first select a set of red edges (via strategy A or B). Then we apply the sequence of barycenter methods as described. Next, we perform the local search, %replacing red edges, optimizing using the sequence of barycenter methods, and keeping the result only if the number of crossings improved. 
and finally, we reorder \textsc{undec} using the barycenter method (if it is non-empty).% one more time, keeping \textsc{out} fixed, while reordering \textsc{undec}.

\section{Evaluation}
\label{sec:evaluation}

We have implemented our AF layout algorithms into an interactive prototype. %\footnote{See footnote~\ref{note:demo} on the title page.}. 
Using this implementation, we quantitatively evaluate our algorithms, showcasing to what extent the exact ILP approach is scalable, and how well the heuristic performs in terms of minimizing crossings. For that, we measure the run time of our algorithms, as well as the number of crossings they achieve on a large set of AFs. Before elaborating on these results, we explain our experimental setup and the data used in our experiments.

\mypar{Experimental setup} We have implemented our heuristics in Javascript V8 (12.4.254.15) and run in Chrome (120.0.6099.217). The ILP model was built in Python 3.8.3 using gurobipy 10.0.3. We ran our evaluation on a machine with an Intel(R) Core(TM) i7-9700K CPU (3.60 GHz, 8 physical cores and 8 logical processors), and 16GB RAM.

For our heuristic pipeline, we observe that red-edge-selection strategy A often slightly outperforms strategy B in crossing minimization. Still, the choice of strategy does not heavily influence the end result, as the local search may later adapt the red edges. Due to space constraints, we evaluate only the variant of our pipeline employing strategy~A.

\begin{figure*}
    \centering
    \includegraphics[width=\textwidth]{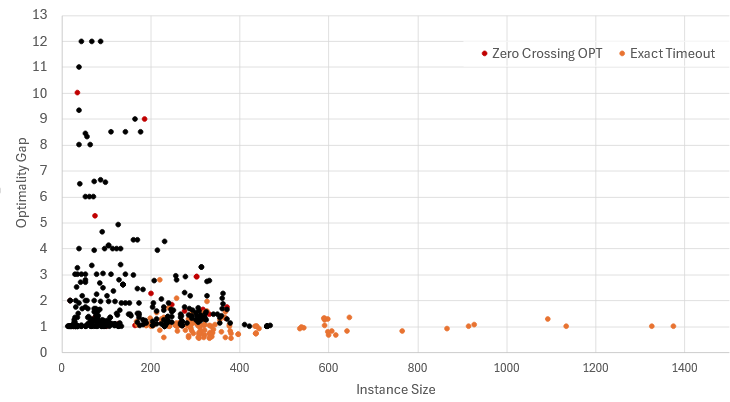}
    \caption{Optimality gap vs. instance size for 546 instances. Black and orange points show the optimality gap as a ratio, while red points show the absolute number of crossings for the heuristic. For orange points, the ILP hit the 30 minute time-out limit.}
    \label{fig:optimality-ratio}
\end{figure*}

\mypar{Data} We used standard AF benchmarks from the ICCMA competition~\cite{DBLP:conf/kr/JarvisaloLN23} with light adaptations: To create a data set with a variety of characteristics (more/fewer arguments in \textsc{in}/\textsc{out}/\textsc{undec}, higher/lower density of attacks between or within certain layers) we removed vertices/edges uniformly at random from larger instances until the desired thresholds for number of vertices and/or edges were achieved. This resulted in 460 2-layer AF drawings and 86 3-layer AF
drawings. The 2-layer instances vary from a minimum of 15 to a maximum of 940
total arguments and attacks. The 3-layer instances range from 27 to 1,374 total arguments and attacks.

\mypar{Measurements} Table~\ref{tab:runtimes} shows the run times measured in our experiments, and Fig.~\ref{fig:optimality-ratio} shows the relation between the instance size and the ratio between the number of crossings produced by our heuristic pipeline, and the number of crossings in solutions of our ILP. In Fig.~\ref{fig:optimality-ratio} orange points indicate that the 30 minute time-out limit was reached by the ILP. As stated in the paper: The ILP may not yet have found an optimal solution for these instances, and the heuristic will therefore sometimes outperform the ILP, leading to an optimality ratio lower than one. Red points indicate that the optimal solution found by the ILP has no crossings. Red points therefore show the absolute number of crossings in the heuristic solution, instead of the (undefined) ratio.

\begin{table}[b]
    \caption{Average run times (in ms) for our exact ILP approach and our heuristic pipeline on instances of various sizes. Instance size is measured by the sum of total arguments and attacks.}
    \label{tab:runtimes}
    \centering
    \begin{tabular}{r|rr}
    \toprule
    & \multicolumn{1}{c}{ILP} & \multicolumn{1}{c}{Heuristic} \\
    \midrule
    0--100 & 8137 & 1.0 \\
    101--300 & 728595 & 8.4\\
    $>$300 & 1222331 & 29.8\\
    \bottomrule
    \end{tabular}
\end{table}

\mypar{Results} %Our results can be found in Table~\ref{tab:runtimes} and Fig.~\ref{fig:optimality-ratio}. 
The run times measured in our evaluation show that the heuristic pipeline is very fast, even for very large instances (at least 300 combined arguments and attacks), where the heuristic on average takes less than 30ms. Conversely, the ILP approach is much slower, taking about 8s on average for small instances (at most 100 arguments and attacks), while large instances require more than 20 minutes on average to find a solution. In our experiments we even see the ILP regularly hit our 30 minute time-out limit. 

We also measured the ratio between the number of crossings in solutions of the heuristic and the number of crossings in the solutions of our ILP. In theory, this ratio should be larger than one. However, in practice, when the time-out limit is hit, the ILP may not yet have found an optimal solution, and the heuristic will therefore sometimes outperform the ILP. Overall, we see that for the large majority of the instances, the heuristic produces at most twice as many crossings as the optimum solution.% of the optimal solution.
We conclude that the heuristic pipeline is an attractive alternative to the exact methods, especially in interactive settings. % such as our prototype implementation.

\section{Conclusion}
%In this work, we proposed 
We proposed novel methods for the visualization of AFs. % and their extensions.
We employ techniques from graph drawing to optimize the visual clarity and support the user to explore and understand the structure of an extension. Future work will include user studies in order to understand to which extent our approach increases the understanding 
of argumentation semantics.

%\section*{Acknowledgments}
%
%The acknowledgements and references fall outside 4-page limit.

% \appendix
% \section{Appendix or supplementary material?}
% Appendices should come before the references, and it is currently unclear to me whether they count towards the page limit. \textbf{The way I read the cfp, they do count towards the 4 page limit} It seems from the order (after acknowledgements) that they do not count. There is also the option to submit supplementary material in a separate file. If we need appendices (for example for the LP formulation) we should see what makes most sense.

%% The file kr.bst is a bibliography style file for BibTeX 0.99c
\bibliographystyle{kr}
\bibliography{references}

\newpage
\appendix
\section{ILP description}\label{app:ILP}
The full description of our ILP can be found on the last page. Line (1) shows the optimization function, and lines (2)-(21) ensure that the crossing variables~$c^n_{ijkl}$, for $n\in\{1,2,3,4\}$ are set correctly. We split the edges in the drawing into four sets: edges between \textsc{in} and \textsc{out} ($E_1$), edges between arguments in \textsc{out} ($E_2$), edges between \textsc{out} and \textsc{undec} ($E_1$), edges between arguments in \textsc{undec} ($E_2$). Lines (2)-(3) and lines (4)-(5) deal with crossings of edges in $E_1$ and $E_3$, respectively, while lines (6)-(13) and (14)-(21) deal with crossings of edges in $E_2$ and $E_4$, respectively.

We also distinguish between the order variables per layer: for \textsc{in} we have variables~$x_{ij}$ which model that argument~$i$ is placed before argument~$j$ in the top-to-bottom order in \textsc{in}. Similarly, we use crossing variables~$y_{ij}$ and~$z_{ij}$ for \textsc{out} and \textsc{undec}, respectively.

The indices we use for arguments are all unique and consecutive: Arguments in \textsc{in} are indexed with $0,\ldots, |\textsc{in}|$, arguments in \textsc{out} have indices $|\textsc{in}|+1,\ldots, |\textsc{in} + \textsc{out}|$, and \textsc{undec} uses $|\textsc{in} + \textsc{out}|+1,\ldots, |\textsc{in} + \textsc{out} + \textsc{undec}|$. 

Line (22) ensures that for each argument $j$ in \textsc{out} exactly one incoming edge is selected as a red edge. For example the edge originating from $i$ can be chosen for argument~$j$, by setting $r_ij = 1$. Line (23) ensures that no two red edges cross, and hence that the solution adheres to the REC.

Lines (24)-(26) are so-called transitivity constraints, which ensure that the order variables model a total and acyclic ordering of the arguments in the respective layers. Lines (27)-(29) make sure the order variables do not contradict each other. Finally, Line (30) restricts variables to 0-1.

\begin{figure*}
\begin{align}
 &\text{minimize} \; \sum\limits_{(i,j),(k,l) \in  E_n} w_{ijkl}^n  \cdot c_{ijkl}^n & \\
 &  \text{s.t.} \; x_{ik} + y_{lj} - c_{ijkl}^1 \leq  1 & (i,j),(k,l) \in E_1 & \\
 &  x_{ki} + y_{jl} - c_{ijkl}^1 \leq  1& (i,j),(k,l) \in E_1 & \\
 & y_{ik} + z_{lj} - c_{ijkl}^3 \leq  1 & (i,j),(k,l) \in E_3 & \\
 &  y_{ki} + z_{jl} - c_{ijkl}^3 \leq  1& (i,j),(k,l) \in E_3 & \\
 &  y_{ik} + y_{kj} + y_{jl} - c_{ijkl}^2 \leq 2 & (i,j),(k,l) \in E_2\\
 &  y_{il} + y_{lj} + y_{jk} - c_{ijkl}^2 \leq 2 & (i,j),(k,l) \in E_2\\
 &  y_{jk} + y_{ki} + y_{il} - c_{ijkl}^2 \leq 2 & (i,j),(k,l) \in E_2\\
 &  y_{jl} + y_{li} + y_{ik} - c_{ijkl}^2 \leq 2 & (i,j),(k,l) \in E_2\\
 &  y_{ki} + y_{il} + y_{lj} - c_{ijkl}^2 \leq 2 & (i,j),(k,l) \in E_2\\
 &  y_{kj} + y_{jl} + y_{li} - c_{ijkl}^2 \leq 2 & (i,j),(k,l) \in E_2\\
 &  y_{lj} + y_{jk} + y_{ki} - c_{ijkl}^2 \leq 2 & (i,j),(k,l) \in E_2\\
 &  y_{li} + y_{ik} + y_{kj} - c_{ijkl}^2 \leq 2 & (i,j),(k,l) \in E_2\\
 &  z_{ik} + z_{kj} + z_{jl} - c_{ijkl}^4 \leq 2 & (i,j),(k,l) \in E_4\\
 &  z_{il} + z_{lj} + z_{jk} - c_{ijkl}^4 \leq 2 & (i,j),(k,l) \in E_4\\
 &  z_{jk} + z_{ki} + z_{il} - c_{ijkl}^4 \leq 2 & (i,j),(k,l) \in E_4\\
 &  z_{jl} + z_{li} + z_{ik} - c_{ijkl}^4 \leq 2 & (i,j),(k,l) \in E_4\\
 &  z_{ki} + z_{il} + z_{lj} - c_{ijkl}^4 \leq 2 & (i,j),(k,l) \in E_4\\
 &  z_{kj} + z_{jl} + z_{li} - c_{ijkl}^4 \leq 2 & (i,j),(k,l) \in E_4\\
 &  z_{lj} + z_{jk} + z_{ki} - c_{ijkl}^4 \leq 2 & (i,j),(k,l) \in E_4\\
 &  z_{li} + z_{ik} + z_{kj} - c_{ijkl}^4 \leq 2 & (i,j),(k,l) \in E_4\\
 &  \sum\limits_{(i,j) \in Att_j} r_{ij} = 1  &  |\textsc{in}| < j \leq |\textsc{in}| + |\textsc{out}|   \\
 &  r_{ij} + r_{kl} + c_{ijkl}^1 \leq 2  &  (i,j),(k,l) \in E_{\textsc{in} \rightarrow \textsc{out}} \\
 &  0 \leq x_{ij} + x_{jk} - x_{ik} \leq 1 & 1 \leq i < j < k \leq |\textsc{in}| & \\
 &  0 \leq y_{ij} + y_{jk} - y_{ik} \leq 1 & |\textsc{in}| < i < j < k \leq |\textsc{in}| + |\textsc{out}| & \\
  &  0 \leq z_{ij} + z_{jk} - z_{ik} \leq 1 & \begin{aligned}
        & |\textsc{in}| + |\textsc{out}| < i < j < k \leq \\
        & |\textsc{in}| + |\textsc{out}| + |\textsc{undec}|
    \end{aligned} & \\
 &  x_{ij} + x_{ji} = 1 & 1 \leq i < j \leq |\textsc{in}| & \\
 &  y_{ij} + y_{ji} = 1 & |\textsc{in}| < i < j \leq |\textsc{in}| + |\textsc{out}| & \\
 & z_{ij} + z_{ji} = 1 & \begin{aligned}
        & |\textsc{in}| + |\textsc{out}| < i < j \leq \\
        & |\textsc{in}| + |\textsc{out}| + |\textsc{undec}|
    \end{aligned} & \\
 & x_{ij},y_{ij},z_{ij},c_{ijkl},r_{ij}  \in \{0,1\}  &
\end{align}
\end{figure*}

\end{document}